\documentclass[conference]{IEEEtran}

\IEEEoverridecommandlockouts
\usepackage[top=.7in, bottom=1.03in, left=.68in, right=.67in]{geometry}

\usepackage{setspace}
\usepackage{graphicx}
\usepackage{epstopdf}
\usepackage{color}

\usepackage[cmex10]{amsmath}
\usepackage{amsthm}
\usepackage{mathtools}
\usepackage{mathptmx} 
\usepackage{times} 
\usepackage{amssymb}
\usepackage{dsfont}
\usepackage{caption}
\usepackage{subcaption}
\usepackage{cite}

\newcommand{\PRP}[1]{\ensuremath{\mathsf{Pr}\left(#1\right)}} 

\begin{document}

\title{Energy Harvesting Wireless Networks with Correlated Energy Sources}

\author{
\IEEEauthorblockN{Mehdi Salehi Heydar Abad}
\IEEEauthorblockA{\normalsize Faculty of Engineering \\
and Natural Sciences\\
Sabanci University,
Istanbul, Turkey\\ mehdis@sabanciuniv.edu}
\and
\IEEEauthorblockN{Deniz Gunduz}
\IEEEauthorblockA{\normalsize Department of Electrical \\
and Electronic Engineering,\\
Imperial College London, U.K.\\ d.gunduz@imperial.ac.uk}
\and
\IEEEauthorblockN{Ozgur Ercetin}
\IEEEauthorblockA{\normalsize Faculty of Engineering \\
and Natural Sciences\\
Sabanci University,
Istanbul, Turkey\\ oercetin@sabanciuniv.edu}
\thanks{This work was in part supported by EC H2020-MSCA-RISE-2015 programme under grant number 690893, by Tubitak under grant number 114E955 and by British Council Institutional Links Program under grant number 173605884.}
}

\maketitle
\newtheorem{theorem}{Theorem}
\newtheorem{lemma}{Lemma}
\newtheorem{corollary}{Corollary}
 
\begin{abstract}
This work considers a system with two energy harvesting (EH) nodes transmitting to a common destination over a random access channel. The amount of harvested energy is assumed to be random and independent over time, but correlated among the nodes possibly with respect to their relative position. A threshold-based transmission policy is developed for the maximization of the expected aggregate network throughput. Assuming that there is no a priori channel state or EH information available to the nodes, the aggregate network throughput is obtained. The optimal thresholds are determined for two practically important special cases: i) at any time only one of the sensors harvests energy due to, for example, physical separation of the nodes; ii) the nodes are spatially close, and at any time, either both nodes or none of them harvests energy.
\end{abstract}

\IEEEpeerreviewmaketitle

\section{Introduction}
Due to the tremendous increase in the number of battery-powered wireless communication devices over the past decade, harvesting of energy from natural resources has become an important research area as a mean of prolonging life time of such devices \cite{Paradiso, Niyato}. The various sources for energy harvesting (EH) are wind turbines, photovoltaic cells, thermoelectric generators and mechanical vibration devices such as piezoelectric devices, electromagnetic devices \cite{EHsource}. EH technology is considered as a promising solution especially for large scale wireless sensor networks (WSNs), where the replacement of batteries is often difficult or cost-prohibitive \cite{Anthony}. However, due to the random nature of the harvested energy from ambient sources, the design of the system requires a careful analysis. In particular, depending on the spatial distribution of EH devices, the amount of energy harvested by different devices is typically correlated. For example, consider EH devices harvesting energy from tidal motion \cite{book}. The locations of two EH devices may be such that one is located at the tidal crest, while the other one is located in a tidal trough. In such a case, there may be a time delay equal to the speed of one wavelength between the generation of energy at each device.

\begin{figure}[ht]
  \centering
    \includegraphics[scale=.2]{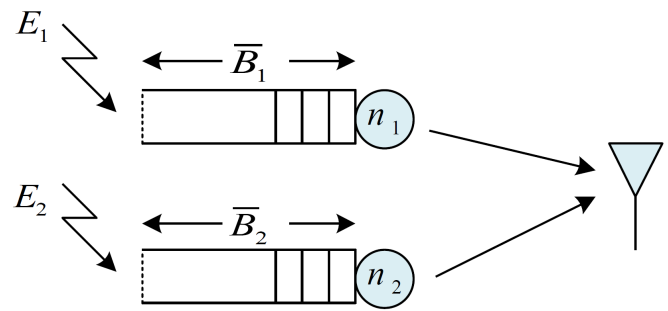}
		  \caption{System Model}
			\label{sysmod}
	\end{figure}

In this paper, we aim to investigate the effects of the correlation between the EH processes at different EH devices in a wireless network. To this end, we consider a network with two EH nodes transmitting data to a common base station over a random access channel as shown in Fig. \ref{sysmod}. Random channel access is a frequently used technique preferred for its distributed and stateless implementation, which is particularly suitable for low power and low duty-cycle sensor networks. In random channel access, the nodes transmit probabilistically over time resulting in occasional packet collisions. However, packet collisions are especially harmful in EH networks due to scarce resources, and should be avoided as much as possible. In this work, we develop and analyze a simple threshold-based transmission policy which grants access to an EH node only when its battery state exceeds a given threshold value. Threshold values are selected based on the battery capacities and the correlation among EH processes of the nodes to maximize the long-term throughput of the system.  

To illustrate the importance of choosing these threshold values intelligently, consider the following example.  Let both EH nodes have a battery capacity of two energy units. Suppose that the EH nodes are spatially close, so they harvest energy simultaneously when energy is available.  If the transmission thresholds are such that both nodes transmit a packet whenever they have one unit of energy, transmissions always result in a collision, and thus, the total network throughput is essentially zero. Meanwhile, if the thresholds are selected such that one EH node transmits a packet whenever it has one unit of energy, and the other node transmits a packet whenever it has two units of energy, there will be a collision once every two transmissions. Hence, with the latter choice of thresholds throughput increases to $0.5$ packets.  

We first derive the average throughput of the network by modeling the system as a discrete time Markov chain (DTMC) and obtaining its steady-state distribution. We then investigate two important special cases to obtain further insights into the selection of optimal transmission thresholds.  In the first special case, only one node harvests energy at any time, while in the second case the nodes always harvest energy simultaneously. These two cases demonstrate completely different optimal threshold characteristics.

Early research in the design of optimal energy management policies for EH networks consider an offline optimization framework \cite{off1,off2}, in which non-causal information on the exact realization of the EH processes are assumed to be available. In the online optimization framework \cite{on1,on2,on3}, the statistics governing the random processes are assumed to be available at the
transmitter, while their realizations are known only causally. The EH communication system is modeled as a  Markov decision process \cite{on1}, and dynamic programming can be used to optimize the throughput numerically. In the learning optimization framework, knowledge about the system behavior is further relaxed and even the statistical knowledge about the random processes
governing the system is not assumed, and the optimal policy
scheduling is learned over time \cite{learning}. In this paper we assume that EH nodes have no knowledge about the EH processes, and can only observe the amount of harvested energy in their own battery. Optimal threshold policies for an EH network is considered in \cite{game} based on a game theoretic approach. In \cite{dos}, authors optimize the throughput of a heterogeneous \emph{ad hoc} EH network by formulating it as an optimal stopping problem. In \cite{basco2015} multiple energy harvesting sensor nodes are scheduled by an access point which does not know the energy harvesting process and battery states of the nodes. However, in these works the EH processes at different devices are assumed to be independent.


\section{System Model}
\label{sec:SystemModel}
We adopt an interference model, where the simultaneous transmissions of two EH nodes result in a collision, and eventual loss of transmitted packets at the base station. Each node is capable of harvesting energy from an ambient resource (solar, wind, vibration, RF, etc.), and storing it in a finite capacity rechargeable battery. EH nodes have no additional power supplies. The nodes are data backlogged, and once they access the channel, they transmit until their battery is completely depleted.  Note that assuming that the nodes are always backlogged allows us to obtain the saturated system throughput. In the following, we neglect the energy consumption due to generation of data to better illustrate the effects of correlated EH processes\footnote{For example, data may be generated by a sensor continuously monitoring the environment. Then, the energy consumption of a sensor may be included as a continuous drain in the energy process, but due to possible energy outages, the data queues may no longer be backlogged. We leave the analysis of this case as a future work.}.

Time is slotted into intervals of unit length. In each time slot, the energy is harvested in units of $\delta$ joules. Let $E_{n}(t)$ be the energy harvested in time slot $t$ by node $n=1,2$.  We assume that $E_{n}(t)$ is an independent and identically distributed (i.i.d.) Bernoulli process with respect to time $t$. However, at a given time slot $t$, $E_{1}(t)$ and $E_{2}(t)$ may not be independent.  
The EH rates are defined as follows:
\begin{align}
\PRP{E_{1}(t)=\delta,E_{2}(t)=\delta}=p_{11},\nonumber\\ 
\PRP{E_{1}(t)=\delta,E_{2}(t)=0}=p_{10}, \nonumber\\
\PRP{E_{1}(t)=0,E_{2}(t)=\delta}=p_{01},\nonumber\\
\PRP{E_{1}(t)=0,E_{2}(t)=0}=p_{00},
\end{align}
where $p_{00}+p_{10}+p_{01}+p_{11}=1$\footnote{Note that if $p_{00}=p_{10}=p_{01}=p_{11}=1/4$, then EH nodes generate energy independently from each other.}.  

We assume that the transmission time $\varepsilon$ is much shorter than the time needed to harvest a unit of energy, i.e., $\varepsilon\ll 1$, and the nodes cannot simultaneously transmit and harvest energy. Transmissions take place at the beginning of time slots, and the energy harvested during time slot $t$ can be used for transmission in time slot $t+1$. The channel is non-fading, and has unit gain. Given transmission power $P$, the transmission rate, $r_n(t)$, $n=1,2$ is given by the Shannon rate, i.e., $r_n(t)=\log\left(1+P/N\right)$ (nats/sec/Hz), where $N$ is the noise power.


We consider a deterministic transmission policy which only depends on the state of the battery of an EH node. Each EH node independently monitors its own battery level, and when it exceeds a pre-defined threshold, the node accesses the channel. If more than one node accesses the channel, a collision occurs and both packets are lost. Note that, by considering such an easy-to-implement and stateless policy, we aim to achieve low-computational power at EH devices. 

The battery of each EH node has a finite capacity of $\bar{B}_{n}$, $n=1,2$. Let $B_n(t)$ be the state of the battery of EH node $n=1,2$ at time $t$. Node $n$ transmits whenever its battery state reaches $\gamma_n\leq \bar{B}_{n}$ joules, $n=1,2$. When node $n$ accesses the channel, it transmits at power $\frac{B_n(t)}{\varepsilon}$, i.e., the battery is completely depleted at every transmission. Hence, the time evolution of the battery states is governed by the following equation.
\begin{align}
B_{n}(t+1)=&\min\left\{\bar{B}_{n},\right. \nonumber\\
&\,\left. B_{n}(t)+E_{n}(t)\mathds{1}_{\left\{B_{n}(t)<\gamma_{i}\right\}}-\mathds{1}_{\left\{B_{n}(t)\geq\gamma_{i}\right\}}B_{n}(t)
\right\}, \label{Bi}
\end{align}
where $\mathds{1}_{a<b}=\begin{cases} 1 & \mbox{if }a<b\\ 0 & \mbox{if } a\geq b \end{cases}$ is the indicator function.  

Let $R_n(t)$ be the rate of {\em successful} transmissions, i.e.,
\begin{align}
R_{1}(t)=&\log\left(1+\frac{B_1(t)/\varepsilon}{N}\right)\mathds{1}_{\left\{B_{1}(t)\geq\gamma_{1},B_{2}(t)<\gamma_{2}\right\}},\label{R1t}\\
R_{2}(t)=&\log\left(1+\frac{B_2(t)/\varepsilon}{N}\right)\mathds{1}_{\left\{B_{1}(t)<\gamma_{1},B_{2}(t)\geq\gamma_{2}\right\}}.\label{R2t}
\end{align}

\section{Maximizing the Throughput}
\label{sec:ThroughputMaximization}

We aim at maximizing the long-term average total throughput by choosing the transmission thresholds intelligently, taking into account the possible correlation between the EH processes.  Let $\bar{R}_n(\gamma_1,\gamma_2)$ be the long-term average throughput of EH node $n$ when the thresholds are selected as $\gamma_1,\gamma_2$, i.e.,
\begin{equation}
\bar{R}_n(\gamma_1,\gamma_2)=\lim_{T\rightarrow\infty}\frac{1}{T} \sum^{T}_{t=1}R_{n}(t), \,\, n=1,2.\label{avgRnt}
\end{equation}
Then, the optimization problem of interest can be stated as
\begin{align}
\max_{\gamma_{1},\gamma_{2}}\   &\sum_{n}{\bar{R}_n(\gamma_1,\gamma_2)},\label{opt1}\\
&\text{s.t.}\ \  \ 1\leq\gamma_{n}\leq \bar{B}_{n}\ \ n=1,2.\label{opt2}
\end{align}

In order to solve the optimization problem (\ref{opt1})-(\ref{opt2}), we first need to determine the long term average total throughput in terms of the thresholds.  Note that for given $\gamma_1,\gamma_2$, the battery states of EH nodes, i.e., $\left(B_{1}(t),\ B_{2}(t)\right)\in \left\{0,\ldots, \gamma_1-1\right\}\times\left\{0,\ldots, \gamma_{2}-1\right\}$ constitute a finite two dimensional discrete-time Markov chain (DTMC), depicted in Fig. \ref{markov}. Let $\pi\left(i,\ j\right)=\PRP{B_{1}(t)=i,\ B_{2}(t)=j}$ be the steady-state distribution of the Markov chain for $i=0,\ldots,\gamma_{1}-1$ and $j=0,\ldots,\gamma_{2}-1$.

\begin{figure}[t]
  \centering
    \includegraphics[scale=.2]{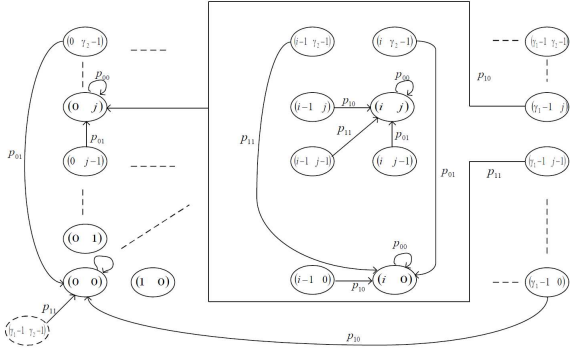}
		  \caption{Associated DTMC with joint battery states}
			\label{markov}
\end{figure}

\begin{theorem}\label{M}
The steady state distribution of DTMC associated with the joint battery state of EH nodes is $\pi\left(i,\ j\right) = \frac{1}{\gamma_{1}\gamma_{2}},$ for $i=0,\ldots,\gamma_{1}-1$ and $j=0,\ldots,\gamma_{2}-1$.
\end{theorem}

\begin{proof}
The detailed balance equations for $i=1,\ldots,\gamma_{1}-1$ and $j=1,\cdots,\gamma_{2}-1$ are:
\begin{align}
\pi\left(i,\ j\right)(1-p_{00})=&\pi\left(i-1,\ j-1\right)p_{11}\nonumber\\
&+\pi\left(i-1,\ j\right)p_{10}+\pi\left(i,\ j-1\right)p_{01}. \label{eq:m1}
\end{align}
Whenever the battery state of node $n$ reaches $\gamma_{n}-1$, in the next state transition, given that it harvests energy, there is a transmission. Since the transmission time is much shorter than a time slot, i.e., $\varepsilon\ll 1$, after reaching state $\gamma_{n}$, node $n$ immediately transmits and transitions back to state $0$. Thus, the detailed balance equations for state $0$ are given as:
\begin{align}
\pi\left(i,\ 0\right)(1-p_{00}) = &\pi\left(i-1,\ 0\right)p_{10}+\pi\left(i,\ \gamma_{2}-1\right)p_{01}\nonumber\\
+&\pi\left(i-1,\ \gamma_{2}-1\right)p_{11},\ \ 1\leq i\leq\gamma_{1}-1\label{eq:m2},
\end{align}
\begin{align}
\pi\left(0,\ j\right)(1-p_{00})= &\pi\left(0,\ j-1\right)p_{01}+\pi\left(\gamma_{1}-1,\ j\right)p_{10}\nonumber\\
+&\pi\left(\gamma_{1}-1,\ j-1\right)p_{11},\ \ 1\leq j\leq\gamma_{2}-1\label{eq:m3},\\
\pi\left(0,\ 0\right)(1-p_{00})&=\pi\left(\gamma_{1}-1,\ \gamma_{2}-1\right)p_{11}\nonumber\\
+&\pi\left(\gamma_{1}-1,\ 0\right)p_{10}+\pi\left(0,\ \gamma_{2}-1\right)p_{01}.\label{eq:m4}
\end{align}
From \eqref{eq:m1}, it is clear that if $p_{01},p_{10}\neq 0$ then $\pi\left(i,\ j\right)\neq 0$ for all $i=1,\ldots,\gamma_{1}-1$ and $j=1,\ldots,\gamma_{2}-1$. Then, it can be verified that $\pi\left(i,\ j\right) = \pi\left(l,\ k\right)$  satisfies (\ref{eq:m1})-(\ref{eq:m4}) for all $i,j,k$, and $l$. Hence, the theorem is proven since $\sum^{\gamma_{2}-1}_{j=0}\sum^{\gamma_{1}-1}_{i=0}\pi\left(i,\ j\right)=1$.
\end{proof}

Once the steady state distribution of DTMC is available, we can obtain the average throughput values. Let $\delta'=\frac{\delta/\varepsilon}{N}$.
\begin{lemma}
\label{lem_avg_thr}
The average throughput of EH nodes 1 and 2 for $p_{01},p_{10}\neq 0$ are given as 
\small
\begin{align}
\bar{R}_{1}\left(\gamma_{1},\gamma_{2}\right)=&\log(1+\gamma_{1}\delta')\nonumber\\
&\times\left(\left(p_{10}+p_{11}\right)\sum^{\gamma_{2}-2}_{j=0}\pi\left(\gamma_{1}-1,\ j\right)+p_{10}\pi\left(\gamma_{1}-1,\ \gamma_{2}-1\right)\right)\nonumber\\
=&\frac{\log(1+\gamma_{1}\delta')\left[(\gamma_{2}-1)\left(p_{10}+p_{11}\right)+p_{10}\right]}{\gamma_{1}\gamma_{2}},\\
\bar{R}_{2}\left(\gamma_{1},\gamma_{2}\right)=&\log(1+\gamma_{2}\delta')\nonumber\\
&\times\left(\left(p_{01}+p_{11}\right)\sum^{\gamma_{1}-2}_{i=0}\pi\left(i,\ \gamma_{2}-1\right)+p_{01}\pi\left(\gamma_{1}-1,\ \gamma_{2}-1\right)\right)\nonumber\\
=&\frac{\log(1+\gamma_{2}\delta')\left[(\gamma_{1}-1)\left(p_{01}+p_{11}\right)+p_{01}\right]}{\gamma_{1}\gamma_{2}}.
\end{align}
\end{lemma}
\normalsize
\begin{proof}
Consider node 1. Note that whenever the batteries are in one of the states $\left(\gamma_{1}-1,\ j\right)$ for $j=0,\ldots,\gamma_{2}-2$, a unit of energy (of $\delta$ joules) is harvested at node $1$ with probability of $p_{10}+p_{11}$, and it transmits in the subsequent transition. Meanwhile, whenever the batteries are in state $\left(\gamma_{1}-1,\ \gamma_{2}-1\right)$, both nodes harvest a unit energy with probability $p_{11}$, and transmit in the subsequent transition resulting in a collision. Thus, in state $\left(\gamma_{1}-1,\ \gamma_{2}-1\right)$, EH node $1$ successfully transmits with probability $p_{10}$. Similar arguments apply for node 2.
\end{proof}


The following optimization problem is equivalent to (\ref{opt1})-(\ref{opt2}).
\begin{align}
\max_{\gamma_{1},\gamma_{2}}\   z(\gamma_{1},\gamma_{2})&\triangleq \frac{\log(1+\gamma_{1}\delta')\left[(\gamma_{2}-1)\left(p_{10}+p_{11}\right)+p_{10}\right]}{\gamma_{1}\gamma_{2}}\nonumber\\
+&\frac{\log(1+\gamma_{2}\delta')\left[(\gamma_{1}-1)\left(p_{01}+p_{11}\right)+p_{01}\right]}{\gamma_{1}\gamma_{2}},\label{optC1}\\
&\text{s.t.}\ \  \ 1\leq\gamma_{n}\leq\bar{B}_{n},\ \ n=1,2.\label{optC2}
\end{align}

Note that (\ref{optC1})-(\ref{optC2}) is an integer program. Since our main motivation is to investigate the effects of the correlated energy arrivals on the operation of EH networks, rather than to obtain exact optimal thresholds, we may relax the optimization problem by omitting the integrality constraints.  Nevertheless, the resulting relaxed optimization problem is still difficult to solve since the objective function is non-convex.  Hence, in the following, we obtain the optimal solution for two important special cases.

\section{Special Cases}\label{sec:SpecialCases}

Depending on the energy source and relative locations of the nodes, correlation among their EH processes may significantly vary. For example, if mechanical vibration is harvested, and the nodes are located far from each other, e.g., one EH device on one side of the road whereas the other one on the other side of a two-lane road, only the EH device on the side of the road where a car passes may generate energy from its vibration.  This is a case of {\em high negative correlation}.  Meanwhile, if solar cells are used as an energy source, EH processes at nearby nodes will have {\em high positive correlation}. 

\subsection{The Case of High Negative Correlation}
\label{sec:HighNegativeCorrelation}

We first analyze the case of high negative correlation. In particular, we have  $p_{00}=p_{11}=0$, $p_{10}=p$ and $p_{01}=1-p$ with $0<p<1$. Note that only one EH device generates energy at a given time.  Let $z^{(-)}\left(\gamma_{1},\gamma_{2}\right)$ be  the total throughput of EH network when the thresholds are $\gamma_{1},\gamma_{2}$, obtained by inserting the values of $p_{00},p_{11},p_{10},p_{01}$ in \eqref{optC1}. We have
\begin{align}
z^{(-)}\left(\gamma_{1},\gamma_{2}\right) = &\frac{\log(1+\gamma_{1}\delta')p}{\gamma_{1}}+\frac{\log(1+\gamma_{2}\delta')(1-p)}{\gamma_{2}}. \label{optN}
\end{align}

The following lemma establishes that an EH device transmits whenever it harvests a single unit of energy.  Interestingly, the optimal thresholds prevent any collisions between transmissions of EH devices, since at a particular time slot only one EH device has sufficient energy to transmit. 
\begin{lemma}
\label{HNC}
The optimal solution of (\ref{optC1})-(\ref{optC2}) when $p_{00}=p_{11}=0$, $p_{10}=p$ and $p_{01}=1-p$ with $0<p<1$, is $\gamma^{*}_{1}=0$, $\gamma^{*}_{2}=0$.
\end{lemma}
\begin{proof}
Assume that $\gamma_{1}$ and $\gamma_{2}$ are non-negative continuous variables. Then, the gradient of $z^{(-)}\left(\gamma_{1},\gamma_{2}\right)$ is:
\begin{align}
\nabla z^{(-)}\left(\gamma_{1},\gamma_{2}\right) = \left[\frac{p \left(\delta'\gamma_{1}-\left(1+\delta'\gamma_{1}\right)\log \left(1+\gamma _1\delta'\right)\right)}{\gamma _{1}{}^2\left(1+\delta'\gamma_{1}\right)}\right.,\nonumber\\
\left.\frac{(1-p) \left(\delta'\gamma_{2}-\left(1+\delta'\gamma_{2}\right)\log \left(1+\gamma _2\delta'\right)\right)}{\gamma _{2}{}^2\left(1+\delta'\gamma_{2}\right)}\right].
\end{align}
Note that $\nabla z^{(-)}\left(\gamma_{1},\gamma_{2}\right)<0$ for all $\gamma_{1}\geq 0$, $\gamma_{2}\geq 0$ and $p$. Since $\nabla z^{(-)}<0$, we have $z^{(-)}\left(\gamma_{1},\gamma_{2}\right)>z^{(-)}\left(\hat{\gamma}_{1},\hat{\gamma_{2}}\right)$ for every $\gamma_{1}<\hat{\gamma}_{1}$ and $\gamma_{2}<\hat{\gamma}_{2}$. Then, the lemma follows.
\end{proof}

\subsection{The Case of High Positive Correlation}
\label{sec:HighPositiveCorrelation}

\begin{figure}[t]
  \centering
    \includegraphics[scale=0.2]{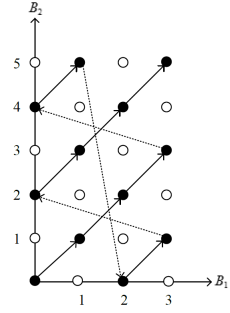}
		  \caption{Transitions of joint battery states for high positive correlation case.}
			\label{HPCex}
\end{figure}

Now, we consider the case of high positive correlation.  In particular, we investigate the optimal solution when EH process parameters are $p_{01}=p_{10}=0$, $p_{11}=p$ and $p_{00}=1-p$ with $0<p<1$; that is, either both EH devices generate energy or neither of them does. Note that in Theorem \ref{M} the steady state distribution of DTMC is derived assuming that all of the states are visited. However, in the case of high positive correlation, only a part of the state space is visited. 

In order to better illustrate this case, consider an EH network with thresholds $\gamma_{1}=4$ and $\gamma_{2}=6$. The state space of the corresponding DTMC is given in Fig. \ref{HPCex}. Large solid and empty circles represent visited and unvisited battery states, respectively. The solid lines represent the transitions of battery states when thresholds are not yet reached, and the dotted lines represent transitions when at least one of the nodes transmits. Also, arrows show the direction of transitions between the states.  Since only a subset of the state space is visited infinitely often, the average throughputs given in Lemma \ref{lem_avg_thr} are no longer valid. We establish the average throughput of EH network with high positive correlation by the following lemma.
\begin{lemma}
The average throughput $\bar{R}_{n}^{(+)}(\gamma_{1},\ \gamma_{2})$ of node $n=1,2$ for $p_{01}=p_{10}=0$, $p_{11}=p$ and $p_{00}=1-p$ is given as 
\begin{align}
\bar{R}_{n}^{(+)}(\gamma_{1},\ \gamma_{2})=&p\cdot\frac{\left[\frac{LCM(\gamma_{1},\ \gamma_{2} )}{\gamma_{n}}-1\right]}{LCM(\gamma_{1},\ \gamma_{2} )}\cdot\log(1+\gamma_{n}\delta'),\,\,n=1,2
\label{eq:avg_thr_positive}
\end{align}
where $LCM(\gamma_{1}, \gamma_{2})$ is the {\rm least common multiple} of $\gamma_{1}$ and $\gamma_{2}$.
\end{lemma}

\begin{proof}
Due to our transmission policy, EH node $n$ transmits whenever its battery level reaches $\gamma_{n}$, $n=1,2$. Note that both nodes reach their respective thresholds simultaneously every $LCM(\gamma_{1},\ \gamma_{2})$ {\em instances of EH events}. Since they transmit simultaneously, a collision occurs, and they both exhaust their batteries, i.e., the joint battery state transitions into state $(0, 0)$. The process repeats afterwards. Hence, the renewal period of this random process is $LCM(\gamma_{1},\ \gamma_{2})$.  In every renewal period, EH node $n=1,2$ makes $\frac{LCM(\gamma_{1},\ \gamma_{2} )}{\gamma_{n}}-1$ number of successful transmissions. Hence, by using renewal reward theory, and noting that on the average a unit of energy is harvested in $p<1$ proportion of time slots, we obtain \eqref{eq:avg_thr_positive}.
\end{proof}


Let $z^{(+)}(\gamma_{1},\ \gamma_{2})=\bar{R}_{1}^{(+)}(\gamma_{1},\ \gamma_{2})+\bar{R}_{2}^{(+)}(\gamma_{1},\ \gamma_{2})$ be the total throughput of a system with high positive correlation. Note that $z^{(+)}(\gamma_{1},\ \gamma_{2})$ is a non-convex function with respect to $\gamma_{1}$, and $\gamma_{2}$. Hence, in the following, we analyze the system in two limiting cases, i.e., when unit of energy harvested per slot, i.e., $\delta'$, is either very small or very large.
 
\subsubsection{Small Values of $\delta'$}
\label{sec:SmallValuesOfDelta}

For small values of $\delta'$, $\log(1+\gamma_{n}\delta')$ can be approximated by $\gamma_{n}\delta'$. Let $GCD(\gamma_1, \gamma_2)$ be the {\em greatest common divisor} of $\gamma_1$ and $\gamma_2$. By substituting $LCM(\gamma_{1},\ \gamma_{2} )=\frac{\gamma_{1}\gamma_{2} }{GCD(\gamma_{1},\ \gamma_{2} )}$ we obtain
\begin{align}
z^{(+)}\left(\gamma_{1},\gamma_{2}\right) &= 2\delta'p-GCD(\gamma_{1},\ \gamma_{2} )\left(\frac{1}{\gamma_{1}}+\frac{1}{\gamma_{2}}\right)\delta'p.\label{optPL}
\end{align}
Note that maximizing (\ref{optPL}) is equivalent to minimizing $GCD(\gamma_{1},\ \gamma_{2} )\left(\frac{1}{\gamma_{1}}+\frac{1}{\gamma_{2}}\right)$.  
Lemma \ref{HPCL} establishes that it is optimal to choose the thresholds as large as possible as long as the greatest common divisor of the two thresholds is equal to $1$. This is due to the fact that the objective function in (\ref{optPL}) is linear, and the optimum thresholds minimize the number of collisions.
\begin{lemma}
\label{HPCL}
The optimal thresholds for the case of high positive correlation for small values of $\delta'$, and for $\bar{B}_{2}>\bar{B}_{1}$  are $\gamma^{*}_{1}=\bar{B}_{1}$, $\gamma^{*}_{2}=\arg\max_{j}\bar{B}_{2}-j$ for $j = 1,\ldots,\bar{B}_{2}$, s.t., $GCD(\bar{B}_{1},j)=1$.
\end{lemma}
\begin{proof}
Note that  $0<\frac{1}{\gamma_{1}}+\frac{1}{\gamma_{2}}\leq 2$, for $1\leq\gamma_n\leq \bar{B}_n$, $n=1,2$. 
Let $\Gamma=\{(\gamma_1,\gamma_2): GCD(\gamma_1,\gamma_2)=1\}$.  Note that if $(\gamma_1,\gamma_2)\notin \Gamma$, then $GCD(\gamma_1,\gamma_2)\geq 2$. Hence, it can be shown that $z^{(+)}\left(\gamma_{1},\gamma_{2}\right)\geq z^{(+)}\left(\gamma_{1}',\gamma_{2}'\right)$, for all $(\gamma_1,\gamma_2)\in \Gamma$, and $(\gamma_1',\gamma_2')\notin \Gamma$. Among $(\gamma_1,\gamma_2)\in\Gamma$, we choose the one that minimizes $\frac{1}{\gamma_{1}}+\frac{1}{\gamma_{2}}$, and thus, proving the lemma.
%
%
\end{proof}

\subsubsection{Large Values of $\delta'$}
\label{sec:LargeValuesOfDelta}
For large values of $\delta'$, $\log(1+\gamma_{n}\delta')$ can be approximated by $\log(\gamma_{n}\delta')$. Also by substituting $LCM(\gamma_{1},\ \gamma_{2} )=\frac{\gamma_{1}\gamma_{2} }{GCD(\gamma_{1},\ \gamma_{2} )}$ in $z^{(+)}(\gamma_{1},\ \gamma_{2})$  we have:
\begin{align}
z^{(+)}\left(\gamma_{1},\gamma_{2}\right) =& \frac{\left(\gamma_{2}-GCD(\gamma_{1},\ \gamma_{2} )\right)\log(\gamma_{1}\delta')p}{\gamma_1\gamma_2}\nonumber\\
&+\frac{\left(\gamma_{1}-GCD(\gamma_{1},\ \gamma_{2} )\right)\log(\gamma_{2}\delta')p}{\gamma_1\gamma_2}.\label{optPH}
\end{align}
The optimal thresholds for this case is established in Lemma \ref{HPCH}. Since the objective function in (\ref{optPH}) has the property of \textit{diminishing returns}, i.e., the rate of increase in the function decreases for higher values of its parameters, each device will choose transmitting more often, equivalently short messages, using less energy.  However, transmissions are scheduled every time each node exceeds a threshold, which dictates small thresholds. When both EH devices transmit with small thresholds, there will be a large number of collisions, so the following lemma suggests that the aggregate throughput is maximized when one EH device transmits short messages, whereas the other transmits long messages.

\begin{lemma}
\label{HPCH}
The optimal thresholds for the case of high positive correlation for large values of $\delta'$ are $\gamma^{*}_{1}=B_{1}$, $\gamma^{*}_{2}=1$ for $\bar{B}_{1}>\bar{B}_{2}$, and they are $\gamma^{*}_{1}=1$, $\gamma^{*}_{2}=B_{2}$ for $\bar{B}_{2}>\bar{B}_{1}$.
\end{lemma}
\begin{proof}
Let ${\hat z}$ be an upper envelope function for $z^{(+)}$, obtained by substituting $GCD(\gamma_{1},\ \gamma_{2} )= 1$ in (\ref{optPH}):
\begin{align}
{\hat z}\left(\gamma_{1},\gamma_{2}\right) = \frac{\left(\gamma_{2}-1\right)\log(\gamma_{1}\delta')p}{\gamma_1\gamma_2}+\frac{\left(\gamma_{1}-1\right)\log(\gamma_{2}\delta')p}{\gamma_1\gamma_2}.
\end{align}

Note that since $GCD(\gamma_{1},\ \gamma_{2} )\geq 1$, for every value of $\gamma_{1}$ and $\gamma_{2}$, we have ${\hat z}\left(\gamma_{1},\gamma_{2}\right)\geq z^{(+)}\left(\gamma_{1},\gamma_{2}\right)$.
First, we maximize ${\hat z}$ for a given $\gamma_{2}$ by obtaining the corresponding optimal $\gamma_{1}$.  Taking the partial derivative of ${\hat z}$ with respect to $\gamma_1$, we obtain:
\begin{align}
\frac{\partial {\hat z} }{\partial \gamma_{1}} = \frac{p}{\gamma _1^2 \gamma _2} \left[ \log \left(\gamma _1 \delta \right)+\log \left(\gamma _2 \delta \right)-\gamma _2 \left(\log \left(\gamma _1 \delta \right)-1\right)-1 \right].
\label{eq:par_derv_zup}
\end{align}
Note that $\gamma_2\in \{1,\ldots, {\bar B}_2\}$.  If $\gamma_2=1$, \eqref{eq:par_derv_zup} reduces to
\begin{align}
\frac{\partial {\hat z}\left(\gamma_{1},1\right) }{\partial \gamma_{1}} =& \frac{p}{\gamma _1^2 \gamma _2}\log\delta>0.
\end{align}

Since $\frac{\partial {\hat z}\left(\gamma_{1},1\right) }{\partial \gamma_{1}}>0$, the maximum value of ${\hat z}$ is attained when $\gamma_1=B_{1}$.
For $\gamma_2=2$, \eqref{eq:par_derv_zup} reduces to
\begin{align}
\frac{\partial {\hat z}\left(\gamma_{1},2\right) }{\partial \gamma_{1}} = &\frac{p}{\gamma _1^2 \gamma _2}\left(-\log \left(\gamma _1 \delta \right)+\log (2 \delta )+1\right)\nonumber\\
&= \left\{
\begin{array}{rl}
<0 & \text{if } \gamma_{1} > 2e,\\
\geq 0 & \text{if } \gamma_{1} \leq 2e,
\end{array} \right.
\end{align}
where $e$ is the Euler's constant. Since $\frac{\partial^{2} {\hat z}\left(2e,2\right) }{\partial \gamma_{1}{}^2}=-\frac{1}{16 e^3}<0$, the maximum value of ${\hat z}$ is attained when $\gamma_1=2e$.  Finally, if $\gamma_2\geq 3$, it can be shown that \eqref{eq:par_derv_zup} is always negative as long as $\delta>3e^2$.
Hence, the maximum value of ${\hat z}$ is attained for $\gamma_1=1$, if $\gamma_{2}\geq 3$.
By comparing the optimal values of $\hat z$ for all $\gamma_2\in \{1,\ldots, {\bar B}_2\}$, one can show that $\hat z$ is maximized for $(\gamma_1,\gamma_2)=\left(B_{1},\ 1\right)$ when $B_1>B_2$ and $(\gamma_1,\gamma_2)=\left(1,\ B_2\right)$ when $B_2>B_1$. Since $GCD(1,\ B_2)=GCD(B_1,\ 1)=1$, and $\hat z=z^{(+)}$ when $GCD(\gamma_1,\gamma_2)=1$, it follows that optimal points for ${\hat z}$ are also the optimal for $z^{(+)}$.
\end{proof}

\section{Numerical Results}
\label{sec:NumericalResults}
We first verify (\ref{optC1}) and (\ref{eq:avg_thr_positive}) by Monte Carlo simulations. In the simulation, we model the battery states using equation (\ref{Bi}). At each time slot $t$, we generate the joint EH process $(E_1(t),E_2(t))$ randomly.  We run the simulation for $10^4$ time slots and calculate the expected throughput by evaluating time average of the instantaneous rates as in (\ref{avgRnt}).

Fig. \ref{figerr} depicts the reliability of our analytical derivations. In particular, we measure both the percent relative error (\%RE), which is defined as $\%\text{RE} = \frac{\text{Analytical value}-\text{Simulation value}}{\text{Analytical value}}\times 100$, and the absolute error (\%AE), which is defined as \%AE = $(\text{Analytical value}-\text{Simulation value})\times 100$, for $\gamma_2=9$ versus $\gamma_{1}$. The results show a good match between the analytical and simulation results.
\begin{figure}[ht]
  \centering
    \includegraphics[scale=.2]{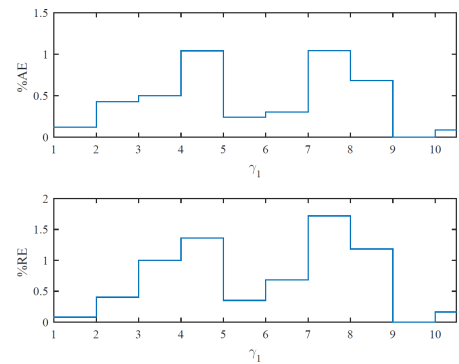}
		  \caption{\%AE and \%RE versus $\gamma_{1}$ with $\gamma_2=9$ and $\delta'=30$.}
			\label{figerr}
\end{figure}

Next, we verify the optimal thresholds by numerically evaluating (\ref{optC1}) and (\ref{eq:avg_thr_positive}) for the cases of high negative and high positive correlation. We assume that $\bar{B}_{1}=\bar{B}_{1}=10$ and $p=0.5$. The aggregate throughput of the network with respect to the thresholds $\gamma_{1}$ and $\gamma_{2}$	for the case of high negative correlation is depicted in Fig. \ref{figHNC}. It can be seen that the optimal thresholds are $\gamma^{*}_{1}=1$, $\gamma^{*}_{2}=1$, which is in accordance with Lemma \ref{HNC}.	
\begin{figure}[ht]
  \centering
    \includegraphics[scale=.3]{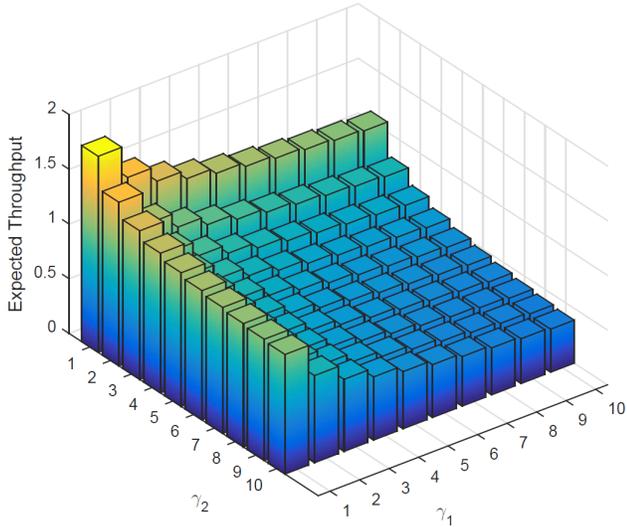}
		  \caption{Expected total throughput for high negative correlation with $\delta'=5$.}
			\label{figHNC}
\end{figure}

Fig. \ref{figHPCL} illustrates the aggregate throughput of the network for the case of high positive correlation with respect to $\gamma_{1}$ and $\gamma_{2}$ for $\delta'=0.04$. The abrupt drops in the value of the aggregate throughput are due to the fact that $GCD(\gamma_{1},\ \gamma_{2})$ varies at least by a factor of two,  which shows consistency with Lemma \ref{HPCL}.

\begin{figure}[ht]
  \centering
    \includegraphics[scale=.3]{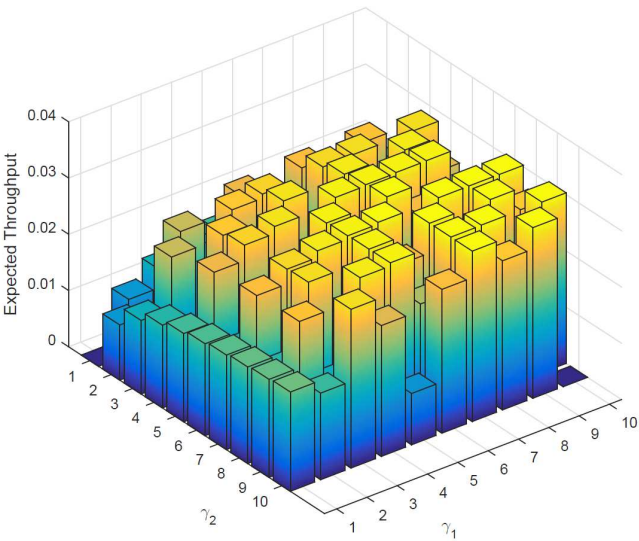}
		  \caption{Expected total throughput for high positive correlation with $\delta'=0.04$.}
			\label{figHPCL}
\end{figure}

In Fig. \ref{figHPCH}, the aggregate throughput is depicted for the case of high positive correlation with respect to $\gamma_{1}$ and $\gamma_{2}$ for $\delta'=30$. As expected from the results established in Lemma \ref{HPCH}, the optimal thresholds are either $(\gamma^{*}_{1},\ \gamma^{*}_{2})=(1,\ 10)$ or $(\gamma^{*}_{1},\ \gamma^{*}_{2})=(10,\ 1)$.

\begin{figure}[ht]
  \centering
    \includegraphics[scale=.3]{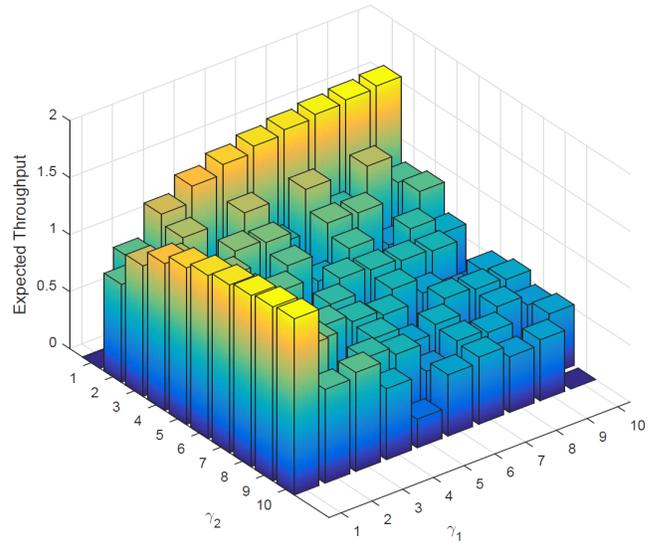}
		  \caption{Expected total throughput for high positive correlation with $\delta'=30$.}
			\label{figHPCH}
\end{figure}

\section{Conclusion}
\label{sec:Conclusion}
We have investigated the effects of correlation among the EH processes of different EH nodes as encountered in many practical scenarios. We have developed a simple threshold based transmission policy to coordinate EH nodes' transmissions in such a way to maximize the long-term aggregate throughput of the network. In the threshold policy, nodes have no knowledge about each other, and at any given time they can only monitor their own battery levels. Considering various assumptions regarding the EH statistics and the amount of the harvested energy, the performance of the proposed threshold policy is studied. The established lemmas in Section \ref{sec:ThroughputMaximization} show that different assumptions about the underlying EH processes and the amount of the harvested energy demonstrate completely different optimal threshold characteristics. As our future work, we will investigate the cases when data queues are not infinitely backlogged and when the channels exhibit fading properties.

\bibliographystyle{IEEEtran} 

\bibliography{Bibliography}

\begin{thebibliography}{10}
\providecommand{\url}[1]{#1}
\csname url@samestyle\endcsname
\providecommand{\newblock}{\relax}
\providecommand{\bibinfo}[2]{#2}
\providecommand{\BIBentrySTDinterwordspacing}{\spaceskip=0pt\relax}
\providecommand{\BIBentryALTinterwordstretchfactor}{4}
\providecommand{\BIBentryALTinterwordspacing}{\spaceskip=\fontdimen2\font plus
\BIBentryALTinterwordstretchfactor\fontdimen3\font minus
  \fontdimen4\font\relax}
\providecommand{\BIBforeignlanguage}[2]{{%
\expandafter\ifx\csname l@#1\endcsname\relax
\typeout{** WARNING: IEEEtran.bst: No hyphenation pattern has been}%
\typeout{** loaded for the language `#1'. Using the pattern for}%
\typeout{** the default language instead.}%
\else
\language=\csname l@#1\endcsname
\fi
#2}}
\providecommand{\BIBdecl}{\relax}
\BIBdecl

\bibitem{Paradiso}
J.~Paradiso and T.~Starner, ``Energy scavenging for mobile and wireless
  electronics,'' \emph{Pervasive Computing, IEEE}, vol.~4, no.~1, pp. 18--27,
  Jan 2005.

\bibitem{Niyato}
D.~Niyato, E.~Hossain, M.~Rashid, and V.~Bhargava, ``Wireless sensor networks
  with energy harvesting technologies: a game-theoretic approach to optimal
  energy management,'' \emph{Wireless Communications, IEEE}, vol.~14, no.~4,
  pp. 90--96, August 2007.

\bibitem{EHsource}
\BIBentryALTinterwordspacing
``Energy harvesting for structural health monitoring sensor networks,''
  \emph{Journal of Infrastructure Systems}, vol.~14, no.~1, pp. 64--79, 2008.
  [Online]. Available:
  \url{http://dx.doi.org/10.1061/(ASCE)1076-0342(2008)14:1(64)}
\BIBentrySTDinterwordspacing

\bibitem{Anthony}
D.~Anthony, W.~Bennett, M.~Vuran, M.~Dwyer, S.~Elbaum, A.~Lacy, M.~Engels, and
  W.~Wehtje, ``Sensing through the continent: Towards monitoring migratory
  birds using cellular sensor networks,'' in \emph{Information Processing in
  Sensor Networks (IPSN), 2012 ACM/IEEE 11th International Conference on},
  April 2012, pp. 329--340.

\bibitem{book}
J.~Trinnaman and A.~Clarke, \emph{2004 Survey of energy resources}.\hskip 1em
  plus 0.5em minus 0.4em\relax Elsevier, 2004.

\bibitem{off1}
M.~Antepli, E.~Uysal-Biyikoglu, and H.~Erkal, ``Optimal packet scheduling on an
  energy harvesting broadcast link,'' \emph{Selected Areas in Communications,
  IEEE Journal on}, vol.~29, no.~8, pp. 1721--1731, September 2011.

\bibitem{off2}
\BIBentryALTinterwordspacing
B.~Devillers and D.~G{\"{u}}nd{\"{u}}z, ``A general framework for the
  optimization of energy harvesting communication systems with battery
  imperfections,'' \emph{CoRR}, vol. abs/1109.5490, 2011. [Online]. Available:
  \url{http://arxiv.org/abs/1109.5490}
\BIBentrySTDinterwordspacing

\bibitem{on1}
Z.~Wang, A.~Tajer, and X.~Wang, ``Communication of energy harvesting tags,''
  \emph{Communications, IEEE Transactions on}, vol.~60, no.~4, pp. 1159--1166,
  April 2012.

\bibitem{on2}
A.~Aprem, C.~Murthy, and N.~Mehta, ``Transmit power control policies for energy
  harvesting sensors with retransmissions,'' \emph{Selected Topics in Signal
  Processing, IEEE Journal of}, vol.~7, no.~5, pp. 895--906, Oct 2013.

\bibitem{on3}
J.~Lei, R.~Yates, and L.~Greenstein, ``A generic model for optimizing
  single-hop transmission policy of replenishable sensors,'' \emph{Wireless
  Communications, IEEE Transactions on}, vol.~8, no.~2, pp. 547--551, Feb 2009.

\bibitem{learning}
\BIBentryALTinterwordspacing
P.~Blasco, D.~G{\"{u}}nd{\"{u}}z, and M.~Dohler, ``A learning theoretic
  approach to energy harvesting communication system optimization,''
  \emph{CoRR}, vol. abs/1208.4290, 2012. [Online]. Available:
  \url{http://arxiv.org/abs/1208.4290}
\BIBentrySTDinterwordspacing

\bibitem{game}
N.~Michelusi and M.~Zorzi, ``Optimal random multiaccess in energy harvesting
  wireless sensor networks,'' in \emph{Communications Workshops (ICC), 2013
  IEEE International Conference on}, June 2013, pp. 463--468.

\bibitem{dos}
H.~Li, C.~Huang, P.~Zhang, S.~Cui, and J.~Zhang, ``Distributed opportunistic
  scheduling for energy harvesting based wireless networks: A two-stage probing
  approach,'' \emph{Networking, IEEE/ACM Transactions on}, vol.~PP, no.~99, pp.
  1--14, 2015.

\bibitem{basco2015}
P.~Blasco and D.~Gunduz, ``Multi-access communications with energy harvesting:
  A multi-armed bandit model and the optimality of the myopic policy,''
  \emph{Selected Areas in Communications, IEEE Journal on}, vol.~33, no.~3, pp.
  585--597, March 2015.

\end{thebibliography}
\end{document}